\documentclass[journal,  onecolumn, 12pt]{IEEEtran} 

\usepackage{amsmath,amsthm,amssymb}
\usepackage{amsfonts}
\usepackage{subfigure}
\usepackage{graphicx}
\usepackage{setspace}

\usepackage{csquotes}
\usepackage{array}

\newtheorem{Lemma}{Lemma}
\newtheorem{Theorem}{Theorem}

\newtheorem{Conjecture}{Conjecture}
\newtheorem{Definition}{Definition}

\def\e{\mathbb{ E}}

\newtheorem{Question}{Question}
\usepackage{color}

\def\bit{\bibitem}

\def\p{\partial}
\def\d{{\rm d}}

\usepackage{extarrows}

\let\ifexpand\iffalse

\begin{document}


\title{A Reformulation of Gaussian Completely Monotone Conjecture:  A Hodge Structure on the Fisher Information along Heat Flow  }
%
%
%

\author{Fan Cheng

\thanks{F.  Cheng is with the Department of Computer Science and Engineering,  Shanghai Jiao Tong University,  Shanghai,  China.  Email: chengfan85@gmail.com.}


}

\maketitle

\begin{abstract}
In the past decade, J. Huh solved several long-standing open problems on log-concave sequences in combinatorics.
The ground-breaking techniques developed in those work are from algebraic geometry:  ``We believe that behind any log-concave sequence
that appears in nature there is such a Hodge structure responsible for the log-concavity''.

A function is called completely monotone if its derivatives alternate in signs; e.g., $e^{-t}$.
A fundamental conjecture in mathematical physics and Shannon information theory is on the complete monotonicity of Gaussian distribution (GCMC), which states that $I(X+Z_t)$\footnote{The probability density function of $X+Z_t$ is called ``heat flow'' in mathematical physics.} is completely monotone in $t$, where $I$ is Fisher information, random variables $X$ and $Z_t$ are independent and $Z_t\sim\mathcal{N}(0,t)$ is Gaussian.

Inspired by the algebraic geometry method introduced by J. Huh,   GCMC is reformulated in the form of a log-convex sequence. In general, a completely monotone function can admit a log-convex sequence and a log-convex sequence can further induce a log-concave sequence.  The new formulation may guide GCMC to the marvelous temple of algebraic geometry.  Moreover,  to make GCMC more accessible to researchers from both information theory and mathematics\footnote{The author was not familiar with algebraic geometry.   The paper is also aimed at providing people outside information theory of necessary background on the history of GCMC in theory and application.},  together with some new findings,  a  thorough summary of the origin, the implication and further study on GCMC is presented.%


\end{abstract}




\section{Introduction}
Gaussian distribution is at the heart of science and engineering,  as it manifests itself in the law of large numbers \cite{Feller68book}, Brownian motion \cite{Karlin81book},  maximum entropy principle \cite{Cover06book}  and several other fundamental results,  serving as the basic building block of information theory,  probability theory and statistics.  In partial differential equations and mathematical physics,  Gaussian distribution has an alias as ``heat kernel'' \cite{Stein03book} for it is the fundamental solution to heat equation.

Recall that,  the heat equation \cite{Evans2010book} has the following form:
 \begin{equation} \notag
u_t  - a \Delta u = 0,
\end{equation}
subject to appropriate initial and boundary conditions.  Here $t>0$, $x\in R^n$ and $a$ is a scalar which is usually equal to $1$. The unknown is $u: R^n\times [0,\infty]\to R$, $u=u(x,t)$,   and the Laplacian $\Delta$ is taken with respect to the spatial variables $x=(x_1, ..., x_n)$: $$\Delta u = \Delta_x u =  \sum_{i=1}^{n}u_{x_ix_i}.$$  When $a=1$,  its fundamental solution is

\[
\Phi(x,t) :=
\begin{cases}
     \frac{1}{(4\pi t)^{n/2}} e^{-\frac{|x|^2}{4t}}  & x  \in R^n,  t>0    \\
    0,              &  x \in R^n, t<0
\end{cases}
\]
If we further consider the initial value (or Cauchy) problem:
\[
\begin{cases}
   u_t-\Delta u   = 0  &  \text{in } R^n\times (0, \infty)    \\
    u                   =g             & \text{on } R^n \times \{t=0\}
\end{cases}
\]
then the solution of initial value problem is the convolution
$$u(x,t) =\int_{R^n} \Phi(x-y, t)g(y)dy = \frac{1}{(4\pi t)^{n/2}} \int_{R^n} e^{-\frac{|x-y|^2}{4t}} g(y) dy  $$

If we restrict to the case where $n=1$,
$$u(x,t)_{n=1} =  \frac{1}{(4\pi t)^{1/2}} \int_{R} e^{-\frac{|x-y|^2}{4t}} g(y) dy.$$
$u(x,t)_{n=1}$ has an practical interpretation  in information theory as it is precisely the probability density function of the received information in the Gaussian Channel.

In information theory \cite{Yeung08book},   Gaussian distribution with mean $\mu$ and  variance $\sigma^2$ is denoted by $\mathcal{N}(\mu, \sigma^2)$. Its probability density function is
$$\mathcal{N}(\mu,\sigma^2) = \frac{1}{\sqrt{2\pi\sigma^2}}e^{-\frac{(x-\mu)^2}{2\sigma^2}}$$
Thus, the heat kernel in dimension 1 can be rewritten as  $N(0,2t) = \sqrt{2t}\mathcal{N}(0,1)$.  If we further restrict $g$ to be a probability density function, then $u(x,t)$ will also be a probability density function.

In a real life communication scenario,  a sender  wants to send a piece of message, denoted by a random variable $X$,  to a receiver.  Due to the existence of environment noise $Z_t\sim\mathcal{N}(0,t)$,  the ``polluted'' message $Y_t$ obtained by the receiver is the sum of $X$ and $Z_t$; i.e.,
$$Y_t = X+Z_t$$
The probability density function $f(y,t)$ of $Y_t$  satisfies $u_t - a\Delta u=0$, where $a= 1/2$.  Thus, heat equation and Gaussian channel can be treated equivalently in essence. $f(y,t)$ is also referred to as the \textit{heat flow} in the literature.

Though information theory and mathematical physics shared the same interest in entropy  after the work of Boltzmann and Shannon, respectively.  The taste is quite divergent.
In mathematical physics,  heat equation is regarded to be quite well-studied and well-understood for hundred years and the attention has shifted to more complicated equations like Boltzmann equation and Landau damping \cite{Villani06book}.  In information theory,  to study a new channel or source model \cite{ElGamal11book},   the special case where the noise is Gaussian is usually discussed first.  The main reason is that, Gaussian distribution matches entropy quite well (e.g., maximum entropy principle) and researchers can obtain some initial results in closed-form which will shed light on the cases with arbitrary noises.  However,  for some complicated communication scenarios in multi-user information theory,  even Gaussian cases are hard to deal with; e.g., Gaussian interference channel \cite{Ahlswede74, Carleial78},    has been notoriously open for over 40 years.   Those long-standing open problems indicates that   a cornerstone of information theory on the Gaussian distribution is missing.


In his seminal 1948 work \cite{Shannon48}, Shannon introduced the entropy power inequality (Shannon EPI) to show that ``Gaussian is the worst additive noise''.  Extensive research had been conducted to explore and generalize Shannon EPI \cite{Stam59, Blachman65, Lieb78, Dembo91,  Villani00,  Rioul11}. Besides its applications in various channel and source coding problems,  Shannon EPI can  be  even used to prove the uncertainty principle and several other fundamental results in mathematical physics \cite{Dembo91}. Shannon EPI is far more than a common inequality in engineering but  a profound characterization of Gaussian distribution in the language of entropy.
Costa \cite{Costa85} in 1985 showed that $\exp(2h(Y_t))$ is concave, where the signs of the first and second derivatives of $h(Y_t)$ were obtained.

In an independent line of work in mathematical physics,  H.  P.  McKean in 1966 \cite{McKean66} asked  that in Boltzmann's $H$-Theorem,   whether the derivatives of $H(u)$ alternate in signs.  His question was finally disproved by E.  Lieb \cite{Lieb82} in 1982.  McKean also studied the higher order derivatives of $h(Y_t)$ with the motivation to generalize the maximal entropy principle to higher order derivatives of $h(Y_t)$.  He showed the signs of the first two derivatives but failed to further show the signs of the third and fourth derivatives.
Surprisingly, the question on heat equation was never studied afterwards and almost lost in the literature~\cite{Villani06book}.

Since 2013, Gaussian completely monotone conjecture (GCMC) \cite{Cheng13}-\cite{ Cheng19}  was introduced in information theory as a purely mathematical problem without any application.  GCMC states that $I(Y_t)$, the Fisher information of $Y_t$, is completely monotone in $t$; i.e.,  the derivatives of  $I(Y_t)$ alternate in signs.  It is almost the same as the question studied by McKean 1966 but stated a  slightly different manner.  However,  compared to McKean's 1966 work,  in \cite{Cheng13}-\cite{ Cheng15}, the signs of the third and fourth derivatives were strictly proved.  In 2019,  the application of GCMC on Gaussian interference channel was explained \cite{Cheng19}. Later, Ledoux, Nair and Wang \cite{Ledoux21} proved the log-convexity conjecture and showed its connection with Gaussian multiuser channels,  which has greatly encouraged the study of GCMC. GCMC may be the right tool to resolve those long-standing open problems with Gaussian noise.
The debut of GCMC also reminded \cite{Toscani15, Ledoux22} the mathematical physics researchers the lost work of McKean. Yet it seems quite tough to prove GCMC \cite{Anan18}.

 The ``gospel'' is from the study of certain log-concave conjectures on combinatorial problems.  J. Huh \cite{Huh14} adopted techniques from algebraic geometry to construct some special algebraic varieties which admit  a log-concave sequence in homology group.   A wide connection between log-concave sequence and Hodge theory \cite{Huh18, Huh22} was introduced to prove log-concave conjectures in a unified manner.
Coincidentally,  completely monotone functions also admit a log-concave characterization.    It will be intriguing to see whether Huh's method work in the setting of completely monotone functions.

The paper is organized as follows. In section II, the notation and terminology is introduced.  The main finding is introduced in section III. In section IV, V and VI, the origin,  the current progress and further implication are presented.

\section{Notation and terminology}
For a random variable $X$,  denote its alphabet (sample space) by $\mathcal{X}$.  If $|\mathcal{X}|$ is finite or infinitely countable,  $X$ is referred to as a discrete random variable \cite{Yeung08book}; otherwise,  it is a continua random variable.

For a discrete random variable $X$ with probability distribution $p(x)$,  its discrete entropy is define by
$$H(X):=\sum_{x\in\mathcal{X}}-p(x)\log p(x)$$
For example,  a binary random variable (Bernoulli distribution) takes values in $\{0, 1\}$ such that $p(X=0)=p$ and $p(X=1)=1-p$,  its entropy is $H(p)=-p\log p -(1-p)\log(1-p) $.


For a continuous random variable $X$ with probability density function (p.d.f) $f(x)$,  its differential entropy\footnote{Without loss of generality, the base of $h$ and $H$ are usually omitted unless otherwise stated.} is define by
$$h(X):=\int_{x}-f(x)\log f(x)\d x$$
The differential entropy of Gaussian distribution is $h(\mathcal{N}(\mu,\sigma^2)) = \frac{1}{2}\log 2\pi e \sigma^2.$

$H$ and $h$ have subtle difference and should not be misused; e.g., $H$ is always non-negative and $h$ can be negative.  For a more detailed treatment on $H$ and $h$,  the readers are referred to  \cite{Yeung08book}.

For ease of notation,
$$f_n := \frac{\d^n}{\d x^n}f(x)$$

\begin{Definition}[Fisher Information]
Let $X$ be a continuous random variable  with density function $f(x)$ and $f(x)$ is differentiable,
$$I(X) = \int \frac{f_x^2}{f}\d x$$
\end{Definition}
\begin{Theorem} [de Bruijn's Identity]
For $Y_t = X+Z_t$, where $X$ and $Z_t\sim \mathcal{N}(0,t)$ are independent,
$$\frac{\d}{\d t}h(Y_t) = \frac{1}{2} I(Y_t).$$
\end{Theorem}

\section{Main Result}

First, we give the definitions of AM and CM functions \cite{Widder46}.
\begin{Definition}[Absolutely monotone function]
A function $f(t)$ is called absolutely monotone (AM) if  the signs of each order derivatives of $f(t)$ are  non-negative; i.e., for $n\geq 0$
$$\frac{\d^n}{\d t^n} f(t) \geq 0$$
\end{Definition}
For example, $e^{t}$ is AM in  $t$.
\begin{Definition}[Completely monotone function]
A function $f(t)$ is called complete monotone (CM) if  all the signs of the derivatives of $f(t)$ alternate in signs; i.e., for $n\geq 0$,
when $n$ is odd,
$$     \frac{\d^n}{\d t^n} f(t) \leq 0;$$
when $n$ is even
     $$\frac{\d^n}{\d t^n} f(t) \geq 0.$$
\end{Definition}
 Both $1/t$ and $e^{-t}$ are CM in $(0, +\infty)$.   AM and CM functions are commonly seen in toy examples.
 It is easy to see if  $f(t)$ is AM,  then $f(-t)$ is CM; and vice versa.  Though to conduct a derivative is machinery, it is  usually non-trivial  to strictly show the  sign.

\begin{Theorem}[Hausdorff-Bernstein-Widder theorem \cite{Widder46}]
A function f(t) is CM iff there exists a non-decreasing and bounded Borel measure $\mu(x)$ on $[0, +\infty)$ such that
$$f(t)=\int e^{-xt}d\mu(x). $$
\end{Theorem}
Noting that the integral above is precisely the Laplace transformation.  In the sequel,  the integral above is referred to as the \textit{Laplace representation}.  The proof \cite{Widder46} of the HBW theorem relies on only elementary calculus,  thus in principle, there is no significant gap between $f(t)$ and its Laplace representation.

\begin{Theorem}[Log-convexity]
If f(t) is CM, then $f(t)$ is log-convex in $t$.
\end{Theorem}
To show $f(t)$ is log-convex,
$$\frac{\d^2}{\d t^2} \log f(t)\geq 0 \Leftrightarrow \frac{f_2f-f_1^2}{f^2}\geq 0 \Leftrightarrow f_1^2 \leq ff_2.$$
It can be proved by Cauchy-Schwartz inequality on the Laplace representation
$$ \left(\int x^2 e^{-xt}\d \mu(x)\right)  \left(\int e^{-xt}\d \mu(x)\right) \geq  \left(\int xe^{-xt}\d \mu(x)\right)^2$$

In the history, it was Hausdorff, Berstein and Widder first studied CM functions.  The 1946 book of Widder \cite{Widder46} is a very good repository on this topic.  The log-convexity theorem can also be found therein (p. 167).  Besides,  Widder also introduced the inequalities induced by the Hankel determinant.  Fink\footnote{Here is an error in my previous paper \cite{Cheng15} on the reference on the log-convexity theorem.} \cite{Fink82} defined a majorization  relation on CM functions and established a set of inequalities including the log-convexity.

%

In \cite{Cheng13}-\cite{ Cheng15},  I studied the following two conjectures.
\begin{Conjecture}[Gaussian Completely Monotone Conjecture]
$I(Y_t)$ is CM in $t$.
\end{Conjecture}
For $1\leq n\leq 4$, the signs of $\frac{\d^n}{\d t^n}h(Y_t)$ have been proved to be $+$, $-$, $+$, and $-$, respectively.
\begin{Conjecture}[Log-convexity conjecture\footnote{When I first noticed the existence of GCMC, I didn't know the pioneering work of Hausdorff-Bernstein-Widder on CM functions.   I just thought that if $I(Y_t)$ is CM, then its convexity could be strengthened.  }]
$I(Y_t)$ is log-convex in $t$.
\end{Conjecture}
The log-convexity conjecture was recently derived from a different viewpoint and  proved in \cite{Ledoux21}, which provides a strong evidence to GCMC.

In the literature of CM functions, there is no systematic way to prove  a function is CM. Besides the proof by definition, to construct $\mu(x)$ to show the complete monotonicity could be found in the proof of BMV conjecture \cite{BMV75}. The work of J.  Huh provides us the possibility to prove complete monotonicity in a unified manner.

First I will introduce some necessary definitions.
\begin{Definition}
A sequence of real numbers $a_0$,  $a_1$,  . . . ,  $a_n$ is said to be log-concave if
$$a_i^2 \geq a_{i-1} a_{i+1}  \text{ for all i}$$
\end{Definition}

A very famous example about log-concave sequence is on chromatic polynomial in algebraic graph theory.
\begin{Definition}[Chromatic polynomial]
For a graph $G$ and $q$ different colors, the number of different colorings such that no adjacent nodes have the same color is denoted by $\chi_G(q)$; e.g.,
\begin{center}
$\chi_G(q)$ = (number of proper colorings of G using q colors),  $q \geq 1$.
\end{center}
\end{Definition}
\begin{Conjecture}[Ronald Read]
The coefficients of the chromatic polynomial form a log-concave sequence for any graph.
\end{Conjecture}
In 2014,  J. Huh introduced a ground-breaking method on these conjectures by applying the techniques in  algebraic geometry; e.g., Hodge theory.   The general idea is  to construct some special algebraic variety such that the characteristics of the homology group is exactly the log-concave sequence in $\chi_G(q)$. A summary on his ideas \cite{Huh17} is quoted  as follows.



\begin{displayquote}
``We will discuss our work on establishing log-concavity of various combinatorial sequences, such as the coefficients of the chromatic polynomial of graphs and the
face numbers of matroid complexes. Our method is motivated by complex algebraic geometry, in particular Hodge theory. From a given combinatorial object $M$ (a matroid),
we construct a graded commutative algebra over the real number
$$A^*(M) = \bigoplus_{q=1}^d A^q(M)$$
which satisfies analogues of Poincar\'e duality, the hard Lefschetz theorem, and the Hodge-Riemann relations for the
cohomology of smooth projective varieties. Log-concavity
will be deduced from the Hodge-Riemann relations for
$M$. We believe that behind any log-concave sequence
that appears in nature there is such a Hodge structure
responsible for the log-concavity.''
\end{displayquote}

Coincidentally,  CM functions also admit a log-convex/log-concave characterization,  which means CM functions may be studied in the viewpoint of algebraic geometry.

\begin{Definition}
A sequence of real numbers $a_0$,  $a_1$,  . . . ,  $a_n$ is said to be log-convex if
$$a_i^2 \leq a_{i-1} a_{i+1}  \text{ for all i}$$
\end{Definition}

By $0 \leq a_i^2 \leq a_{i-1} a_{i+1}$,  $(\frac{1}{a_i})^2 \geq \frac{1}{a_{i-1}} \frac{1}{a_{i+1}}.$
\begin{Lemma}
If a sequence $\{a_i\}$ is log-convex, then $\{\frac{1}{a_i}\}$ is log-concave.
\end{Lemma}
The converse is not true as $(\frac{1}{a_i})^2 \geq \frac{1}{a_{i-1}} \frac{1}{a_{i+1}}$ cannot imply $a_i^2 \leq a_{i-1} a_{i+1}$.  There is a gap between log-convexity and log-concavity in general.


\begin{Lemma}\label{log-lemma}
If $f(t)$ is CM in $t$, then $\{f_i\}$ is a log-convex sequence.
\end{Lemma}
\begin{proof}
A function $f$ is CM, then $f$ is always log-convex.
Furthermore,  if a  function $f$ is CM, then $|f_n|$ is also CM and $\log |f_n|$ is convex.  That is, $f_n^2\leq f_{n-1}f_{n+1}$.  Thus, $\{f_i\}$ is a log-convex sequence.

The result $f_n^2 \leq f_{n-1}f_{n+1}$  can also be proved by  applying Cauchy-Schwartz inequality on the Laplace representation of $f(t)$.
$$ \left(\int x^{n-1} e^{-xt}\d \mu(x)\right)  \left(\int x^{n+1}e^{-xt}\d \mu(x)\right) \geq  \left(\int x^{n}e^{-xt}\d \mu(x)\right)^2$$
\end{proof}

The converse is not always true.  However, if $\{f_i\}$ is a log-convex sequence,  then $f_if_{i+2}\geq 0$ always holds; i.e., the complete monotonicity of $f(t)$ is decided by the signs of $f$ and $f_1$,  which are easier to check than $f_n$.


\begin{Question}
Can we use the algebraic geometry method of J.  Huh to construct CM functions systematically?
\end{Question}
Though algebraic geometry method has a long history in coding theory,  in Shannon theory (channel/source coding problems),  it is unprecedented to make use of tools in algebraic geometry.  Yet, in light of J.  Huh's ground-breaking work, it is very tempting to make such an adventure.
Personally speaking, CM is from analysis but is is indeed algebraic as it is a certain invariant. For GCMC, the invariant is an important characterization of Gaussian distribution. Indeed, heat equation  describes a very fundamental physical phenomena in nature.  


 \section{The origin of GCMC}
There are widely  correspondences between binary random variable and Gaussian random variable in the literature of Shannon information theory.  Similar to Gaussian random noise, if the noise in a source or channel coding problem is Bernoulli distribution,  then usually an expression in  closed-form can also be obtained.  Even the the proofs can be obtained by the same arguments (e.g., superposition coding).   Some of the correspondences are listed as follows.
\begin{center}
\begin{tabular}{ | l | l| }
  \hline
  Binary symmetric channel & Gaussian channel  \\
  \hline
  Binary broadcast channel & Gaussian broadcast channel  \\
  \hline
  Binary MAC & Gaussian MAC \\
  \hline
  Binary Slepian-wolf coding &   Gaussian Slepian-wolf coding \\
  \hline
\end{tabular}
\end{center}
The should be a theory to explain the underlying interplay between binary random variable and Gaussian random variable.

In most situation, to simplify the communication model,  information theory study usually assumes that the size of the alphabet of the random variables are all binary.  In this setting, Mrs.  Gerber's Lemma (MGL) is the most widely used tools.

For $0\leq p,  q \leq 1$,  the convolution of $p$,  $q$ is
$$p*q:=p\times(1-q)+(1-p)\times q.$$

\begin{Lemma}[Mrs. Gerber's Lemma (MGL)]
$H(p*H^{-1}(x))$ is convex in $x$, for $0\leq p \leq 1$.
\end{Lemma}

The study of GCMC originated from the study of Mrs.  Gerber's Lemma.   As $H(p*H^{-1}(x))$ is symmetric in $p=1/2$, we could assume $0\leq p\leq 1/2$.  If we look at the shape of  $H(p*H^{-1}(x))$ from $p$,  we can find that as $p$ goes from $0$ to $1/2$,  there is a deformation on  $H(p*H^{-1}(x))$ from $H(p*H^{-1}(x))$ = $x$ to $H(p*H^{-1}(x))$=$1$.  During the process of deformation, the convexity has been remained.  As binary entropy function is a very ``nice'' function,  one may wonder whether the convexities at $p=0$ and $p=1/2$ suffice to enforce the convexity in  $p\in [0, 1/2]$.  The idea was finally rigorously coined in the following result.


\begin{Theorem} [GMGL] \footnote{C. Nair suggested the study on the convexity of general form  $H(p*g(x))$.}
$H(p*g(x))$, $0 \leq p \leq 1$ is convex in $x$ if and only if $H(g(x))$ is convex in $x$.
\end{Theorem}
The general idea of the proof  is first to take the 2nd derivative with respect to $x$ of $H(p*g(x))$. Then further take the 2nd derivative with respect to $p$, and show that the expression obtained is $\leq 0$; i.e., it is concave in $p$.  By this mean,  GMGL holds if and only if it holds at  $p=0$ and $p=1/2$.

When GMGL is proved,  a question naturally arose from the  correspondence between   binary random variable and Gaussian random variable: what is the counterpart in Gaussian distribution?
This leads to the debut of GCMC.

\section{Gaussian completely monotone conjecture}

Recall that Gaussian channel was first introduced in  Shannon 1948 \cite{Shannon48}: $X$ is the information to send,  the independent noise $Z_t=\mathcal{N}(0,t)$ is  Gaussian,
$$Y_t = X+ Z_t.$$
As we need to restrict the power assumption of the system, there is a power constraint on $X$
$$\e  X^2 \leq P$$
The channel capacity is given by
$$C(S) = \frac{1}{2}\log  (1+\frac{P}{t}) $$

Shannon further asserted that Gaussian  noise is  the worst additive noise  when the variance of the noise is given.   That is, for a general additive channel $Y = X + Z$, where the maximum variance of  $Z$ is given,   the channel capacity is minimized iff $Z$ is Gaussian.  To prove it, Shannon introduced the famous entropy power inequality.

\begin{Theorem}[Shannon Entropy Power Inequality (Shannon EPI)]
For any two independent continuous random variables $X$ and $Y$,
$$e^{2h(X+Y)}\geq e^{2h(X)}+e^{2h(Y)}.$$
The equality holds iff  both $X$ and $Y$ are Gaussian.
\end{Theorem}


As Gaussian distribution is the optimizer,  it is  natural to study $h(X+Z_t)$,  i.e., the Gaussian perturbation of $h(X)$. Its first order derivative was studied almost a century ago.
$$\frac{\partial }{\partial t} h(X+\sqrt{t}Z) =\frac12 I(Y_t),$$
where $I$ denotes the Fisher information aforementioned.

The second derivative of $h(Y_t)$ was obtained by H.  McKean in mathematical physics and Costa in information theory, respectively.

Let $f(y, t)$ be the probability density function of $Y_t$.
\begin{Theorem}
 $$\frac{\p^2 }{\p t^2}h(Y_t) = -\frac{1}{2}\int f \left(\frac{f_{yy}}{f}-\frac{f_y^2}{f^2}\right)^2\d y.$$
\end{Theorem}
 Recall that
$$f_n := \frac{\d^n}{\d y^n}f(y,t)$$
\begin{Theorem}
For  $t>0$,
\begin{align}
\frac{\p^3}{\p t^3}h(Y_t) 
&= \frac{1}{2}\int f\left( \frac{f_{3}}{f}-\frac{f_{1}f_{2}}{f^2}+\frac{1}{3}\frac{f_1^3}{f^3}\right)^2  +  \frac{f_1^6}{45f^5} \d y. \notag \label{eq06666}
\end{align}
This implies that $I(Y_t)$ is convex in $t$.
\end{Theorem}

\begin{Theorem}
For  $t>0$,
\begin{align}
&\frac{\p^4}{\p t^4}h(Y_t)\notag \\
&= -\frac{1}{2}\int f \left( \frac{f_4}{f} -\frac{6}{5} \frac{f_1f_3}{f^2} -\frac{7}{10} \frac{f_2^2}{f^2} +\frac{8}{5} \frac{f_1^2f_2}{f^3} -\frac{1}{2} \frac{f_1^4}{f^4} \right)^2 \notag \\
& \indent  + f \left(\frac{2}{5} \frac{f_1f_3}{f^2}  -\frac{1}{3} \frac{f_1^2f_2}{f^3} + \frac{9}{100} \frac{f_1^4}{f^4}\right)^2 \notag\\
& \indent  + f \left( -\frac{4}{100} \frac{f_1^2f_2}{f^3}  + \frac{4}{100} \frac{f_1^4}{f^4}\right)^2 \notag\\
& \indent  + \frac{1}{300}   \frac{f_2^4}{f^3} +\frac{56}{90000}  \frac{f_1^4f_2^2}{f^5} + \frac{13}{70000}  \frac{f_1^8}{f^7} \d y. \notag
\end{align}
\end{Theorem}

When $X\sim \mathcal{N}(0, a)$,  $h(Y_t) = \frac{1}{2}\log2\pi e (a+t)$ and $I(Y_t)=1/(t+a)$.    It is easy to see $I(Y_t)=1/(t+a)$ is CM in $t$.  Thus,  the signs of the above four derivatives guide us to the Gaussian completely monotone conjecture and the log-convexity conjecture of $I(Y_t)$.

One need to notice that it is machinery to derive any order derivative of a function, however, it is non-trivial to show the signs.
I obtained the third derivative solely by hand and it was first reported in HKU 2013 \cite{Cheng13},  as the evidence to convince the audience of the existence of  GCMC and log-convexity. Afterwards, the 4th derivative was calculated as an extension to 3rd derivative, but  the  expression is too complicated to obtain the coefficients by hand.
The  coefficients  were finally obtained by Y. Geng (the 2nd author in \cite{Cheng14b, Cheng15}) by numerically searching in  the space of feasible solutions via a matlab program implementing the gradient descent algorithm.  The numerical solution was then refined to the rational form above.

The original statement of McKean's problem in 1966 was stronger than GCMC as McKean asked whether Gaussian distribution could achieve the extremality in the derivatives when $n=1, 2, 3, 4, ...$. However, the gap between GCMC and McKean's problem on heat flow is negligible because it is well-known  that ``Gaussian maximize entropy''. If one could solve GCMC, then it will automatically turn to the Gaussian extremality. However, heat flow was not the major concern in McKean's paper.

%



%


\section{GCMC: further Implications}
\subsection{A Useful Test Tool for Gaussian Distribution}
To show the sign of 2nd derivative  of  $I(Y_t)$ is the key progress in the discovery of GCMC,  which has convinced the audience of the existence of GCMC.   The process to find such an expression is full of involved hard calculus. One should be bold enough to tolerate the panic along this way.  It will be very important to have a solid confidence on GCMC.   The secret behind my adventure is that I have the following test tool to verify a  conjecture on Gaussian distribution.  

The test tool is rather simple. Take a mixed Gaussian distribution $X \sim \lambda \mathcal{N}(\mu_1, 1) + (1-\lambda)\mathcal{N}(\mu_2, 1)$, where $\mu_1\leq \mu_2$ and $0 < \lambda < 1$. $X$ can be treated as the average of two Gaussian distributions at $\mu_1$ and $\mu_2$.  Let $d=\mu_2-\mu_1$.  When $d$ is very large, $X$ can be regarded as the union of two Gaussian distributions.  First, put $\mu_1$ very far way from $\mu_2$. When $t$ increases,  the Gaussian noise $Z_t$ will make the Gaussian distribution at  $\mu_1$ move close to  $\mu_2$. Then a counter-example may be discovered if  the proposed generalization of Shannon EPI is wrong.  By this approach,  I have disproved several possible generalizations. However, for GCMC, I numerically checked up to the 7th order and log-convexity, the numerical result is still true.   This is the most important secret tool compared to McKean in 1966.

%
%
%

\subsection{Application of GCMC}
As the motivation of GCMC is not application based,  it is always a challenge on the application of GCMC.  Also compared to the application of  Shannon EPI,  it is hard to directly apply GCMC in the same manner.  Intuitively, GCMC should be very useful by its nice form.  We can find the possible application via the Laplace representation of CM functions.
From the Laplace representation, we can see that $f(t)$ can be divided into two parts; $e^{-xt}$ is the completely monotone part and $\mu(x)$ is the identity of $f(t)$. Those two parts can be separated and  $e^{-xt}$ is the CM part which stopped us from a close observation on $I(Y_t)$.  Also, we can further see that $I(Y_t)$ can be recovered by its $\mu(x)$, which indicates that $I(Y_t)$ can be reversed, which is quite special as we knew in general the ``information'' of the system cannot be reversed by the second law of thermodynamics.

The separation of $e^{-xt}$ and $\mu(x)$ may lead to possible applications in network information theory.  In Shannon information, Shannon himself perfectly solved the channel capacity of discrete memoryless channel.  After that,  even for the channel capacity with one more extra node in the communication model,  we cannot find a satisfied answer till now.  The reason is that after the intermediate node, the information sent and the information obtained has been encoded by the transition probability inside the node.  So far we have no good theory to govern this  encoding process.
The separation of $e^{-xt}$ and $\mu(x)$ has provided such an opportunity  as it tells us that after the encoding,  the information can still be ``recovered'' by $\mu(x)$.  This indicates that for certain channel models with Gaussian noise,  e.g., Gaussian interference channel,  GCMC may have the possibility to help find the capacity region.

Amazingly,  recently,  M.  Ledoux, C. Nair and   Y.  Wang \cite{Ledoux21},  found a result on Gaussian multiuser channel which is equivalent to  the log-convexity of $I(Y_t)$. After further studied on the 2nd derivative expression on $I(Y_t)$, they finally proved it.  It is a very encouraging progress on both GCMC and Gaussian multiuser channels.

\section{Conclusion}

Though the idea introduced in this paper will face many challenges, it  will be a fruitful adventure to explore the relation between Shannon entropy and concepts in algebraic geometry. If it could be established one day, then there will be a fundamental change in the viewpoint of Shannon theory and the applications based on it.


\end{document}